\documentclass{amsart}
\usepackage{mathrsfs}
\usepackage{amssymb}
\usepackage{amsmath}
\usepackage{amsfonts}
\usepackage{enumerate}
\usepackage{pifont}
\usepackage{enumerate}
\usepackage[pdftex]{graphicx}
\usepackage{verbatim}
\usepackage{amsthm}
\usepackage[all]{xy}
\usepackage{xypic}

\numberwithin{equation}{section}
\newtheorem{theorem}{Theorem}[section]

\newtheorem{definition}{Definition}[section]

\newtheorem{proposition}{Proposition}[section]
\newtheorem{remark}{Remark}[section]
\newtheorem{example}{Example}[section]

\begin{document}

\title{Observable-geometric phases and quantum computation}

\author{Zeqian Chen}

\address{State Key Laboratory of Resonances and Atomic and Molecular Physics, Wuhan Institute of Physics and Mathematics, Chinese Academy of Sciences, 30 West District, Xiao-Hong-Shan, Wuhan 430071, China.}

\thanks{Key words: Geometric phase; Observable-geometric phase; Heisenberg equation; quantum connection, quantum computation.}

\date{}
\maketitle
\markboth{Zeqian Chen}%
{Quantum phase}

\begin{abstract}
This paper presents an alternative approach to geometric phases from the observable point of view. Precisely, we introduce the notion of observable-geometric phases, which is defined as a sequence of phases associated with a complete set of eigenstates of the observable. The observable-geometric phases are shown to be connected with the quantum geometry of the observable space evolving according to the Heisenberg equation. They are indeed distinct from Berry's phase \cite{Berry1984, Simon1983} as the system evolves adiabatically. It is shown that the observable-geometric phases can be used to realize a universal set of quantum gates in quantum computation. This scheme leads to the same gates as the Abelian geometric gates of Zhu and Wang \cite{ZW2002,ZW2003}, but based on the quantum geometry of the observable space beyond the state space.
\end{abstract}



\section{Introduction}\label{Intro}

The geometric phase can be used in quantum computing by implementing each of quantum gates in terms of geometric phases only \cite{ZR1999}. This scheme for computing is called geometric quantum computation. Geometric quantum computation involves adiabatic or nonadiabatic, as well as Abelian or non-Abelian characteristics of the underlying quantum evolution for the quantum state. It is shown in \cite{ZW2002,ZW2003} that the non-adiabatic Abelian geometric phases can be used to realize a universal set of quantum gates and thus are sufficient for universal all-geometric quantum computation. Furthermore, the non-adiabatic non-Abelian geometric phases are shown in \cite{STAHJS2012} to achieve universality as well. We refer to \cite{SMC2016} for a up-date review and background.

Here, we present an alternative approach to the geometric phase from the observable point of view. Precisely, we introduce the concept of observable-geometric phases, which is defined as a sequence of phases associated with a complete set of eigenstates of the observable. The observable-geometric phases are shown to be connected with the geometry of the observable space evolving according to the Heisenberg equation. In particular, the observable-geometric phases obtained are distinct from Berry's phase \cite{Berry1984, Simon1983} as the system evolves adiabatically. As application, the observable-geometric phases can be used to realize a universal set of quantum gates in quantum computation. The quantum gate based on the observable-geometric phases can be interpreted as the Abelian geometric (ZW) gates of Zhu and Wang \cite{ZW2002,ZW2003}; indeed the two schemes lead to the same gates. However, our gates are connected with a quantum geometric structure over the observable space induced by a complete measurement, distinct from that the ZW gates are related with classical (Riemannian) geometry of the state space of the system.

The paper is organized as follows. In Section \ref{ObGeoPhase}, we will define the notion of the observable-geometric phase using the propagator of a time-dependent Hamiltonian, which is shown to be independent of the choice of the Hamiltonian as long as the Heisenberg equations involving the Hamiltonian. Also, we show that the observable-geometric phase does not coincide with Berry's phase as the system evolves adiabatically. In Section \ref{Ex}, the observable-geometric phase of a qubit system is computed in details, and moreover, we will show how to make use of the Lewis-Resenfeld invariant theory \cite{LR1969} computing the observable-geometric phase of a general quantum system. As application, in Section \ref{GQC}, we show that the observable-geometric phases can be used to realize a universal set of quantum gates in quantum computation. We will give a summary in Section \ref{Conclusion}. Finally, we include an appendix as Section \ref{Appendix} on a quantum-geometric structure over the observable space consisting of all complete orthogonal decompositions.

For the sake of simplicity, we only consider finite quantum systems, namely the associated Hilbert spaces $\mathbb{H}$ have a finite dimension. This is sufficient for quantum computation. In what follows, we always denote by $\mathcal{B} (\mathbb{H})$ the algebra of all bounded operators on $\mathbb{H},$ by $\mathcal{O} (\mathbb{H})$ the set of all self-adjoint operators on $\mathbb{H},$ and by $\mathcal{U} (\mathbb{H})$ the group of all unitary operators on $\mathbb{H}.$ Without specified otherwise, the integer $d$ always denotes the dimension of $\mathbb{H},$ and $I$ the identity operator on $\mathbb{H}.$

\section{Observable-geometric phase}\label{ObGeoPhase}

Consider a finite quantum system with a time-dependent Hamiltonian $h(t).$ Since the associated Hilbert space $\mathbb{H}$ has finite dimension $d,$ each Hamiltonian $h(t)$ has a discrete spectrum as well as all observables considered in the following. Indeed, the spectrum of $h(t)$ at any given time will not be of importance. Instead, we shall consider the time evolution operator, as called propagator, generated by $h(t)$ (see \cite[Theorem X.69]{RS1980}). This is a two-parameter continuous family of unitary operators such that for any $t, r, s \in \mathbb{R},$
\begin{equation}\label{eq:Propagator}
U(t,t) = I,\quad U(t, r) U (r, s) = U(t,s),
\end{equation}
and
\begin{equation}\label{eq:SchrodingerEquPropagator}
\mathrm{i} \frac{d}{d t} U(t,s) = h(t) U(t,s),
\end{equation}
that is, for any $s \in \mathbb{R}$ and $\phi \in \mathbb{H},$ $\phi_s (t) = U (t,s) \phi$ is the unique solution of the time-dependent Schr\"{o}dinger equation
\begin{equation}\label{eq:SchrodingerEquTime}
\mathrm{i} \frac{d}{d t} \phi_s(t) = h(t) \phi_s (t)
\end{equation}
with $\phi_s (s) = \phi.$ Given any observable $X_0,$ namely a self-adjoint operator on $\mathbb{H},$ by \eqref{eq:SchrodingerEquPropagator} we conclude that $X(t) = U(0,t) X_0 U(t,0)$ is the unique solution of the time-dependent Heisenberg equation
\begin{equation}\label{eq:HeisenbergEquTime}
\mathrm{i} \frac{d X (t)}{d t} = [X (t), \tilde{h}(t)]
\end{equation}
with $X(0) = X_0,$ where $\tilde{h}(t) = U(0,t) h(t) U(t, 0).$ If there exists $T>0$ such that $X(T) = X(0),$ the time evolution of observable $(X(t): t \in \mathbb{R})$ is then called cyclic with period $T,$ and $X_0 = X(0)$ is said to be a cyclic observable.

Given an observable $X_0$ with non-degenerate eigenstates $\psi_n$ for all $1 \le n \le d,$ assume that $X(t) = U(0,t) X_0 U(t,0)$ is cyclic with period $T,$ namely $X(T) = X_0$ so that $U(0, T)\psi_n = e^{\mathrm{i} \theta_n}\psi_n$ with some $\theta_n \in [0, 2 \pi)$ for every $n =1, \ldots, d.$ Denoted by $\psi_n (t) = U(0,t) \psi_n$ for $1 \le n \le d,$ which are the eigenstates of $X(t),$ we conclude that $\psi_n (t)$ satisfies the skew (time-dependent) Schr\"{o}dinger equation
\begin{equation}\label{eq:SchrodingerEquTimeEigenstate}
\mathrm{i} \frac{d}{d t} \psi_n (t) = - \tilde{h}(t) \psi_n (t)
\end{equation}
with $\psi_n (0) = \psi_n,$ due to the fact that $U(0,t) = U(t,0)^{-1}.$  Define
\begin{equation}\label{eq:Parallelvect}
|\tilde{\psi}_n (t) \rangle = e^{ - \mathrm{i} \int^t_0 \langle \psi_n (s) | \tilde{h} (s) | \psi_n (s) \rangle ds } |\psi_n (t) \rangle
\end{equation}
and
\begin{equation}\label{eq:ClosedCurve}
|\bar{\psi}_n (t) \rangle =  e^{- \mathrm{i} \alpha_n (t)} |\psi_n (t) \rangle
\end{equation}
with $\alpha_n (t) \in [0, 2 \pi)$ being continuous in $t$ such that $\alpha_n (T) - \alpha_n (0) = \theta_n,$ i.e., $|\bar{\psi}_n (T) \rangle = |\bar{\psi}_n (0) \rangle$ for every $n \ge 1.$ Note that for any $n \ge 1,$
$$
|\tilde{\psi}_n (t) \rangle = e^{ - \mathrm{i} \int^t_0 \langle \psi_n (0) | h (s) | \psi_n (0) \rangle ds } |\psi_n (t) \rangle
$$
because $\tilde{h}(t) = U(0,t) h(t) U(t, 0),$ and $|\tilde{\psi}_n (T) \rangle = e^{\mathrm{i} \beta_n}|\psi_n (0) \rangle$ with
$$
\beta_n = \theta_n - \int^T_0 \langle \psi_n (0) | h(t) | \psi_n (0)\rangle d t.
$$

Now, from \eqref{eq:SchrodingerEquTimeEigenstate} we conclude
\begin{equation}\label{eq:ParallelCondVect}
\langle \tilde{\psi}_n (t) | \frac{d}{d t} |\tilde{\psi}_n (t) \rangle =0,
\end{equation}
and
\begin{equation}\label{eq:ObGP-ClosedCurve}
\beta_n = \int^T_0 \mathrm{i} \langle \bar{\psi}_n (t) | \frac{d}{d t} |\bar{\psi}_n (t) \rangle d t
\end{equation}
since $|\bar{\psi}_n (T) \rangle = |\bar{\psi}_n (0) \rangle.$ This leads to the notion of the geometric phase for observable as follows.

\begin{definition}\label{df:ObGeoPhase}
Using the above notations, we define the geometric phases of the periodic evolution of observable $X(t)$ by
\begin{equation}\label{eq:q-GeoPhase}
\beta_n = \theta_n - \int^T_0 \langle \psi_n (0) | h(t) | \psi_n (0)\rangle d t
\end{equation}
which is uniquely defined up to $2 \pi k$ ($k$ is integer) for every $n = 1, \ldots, d.$ We simply call $\beta_n$'s the observable-geometric phases.
\end{definition}

\begin{remark}\rm
\begin{enumerate}[{\rm 1)}]

\item Note that for every $n = 1, \ldots, d,$ $\theta_n = \arg \langle \psi_n (0)| U(T,0) | \psi_n (0)\rangle$ and
$$
\langle \psi_n (0) | h(t) | \psi_n (0)\rangle = \langle \psi_n (t) | \tilde{h}(t) | \psi_n (t) \rangle
$$
for all $t.$ Thus, the observable-geometric phases $\beta_n$'s are obtained by removing the dynamical part from the total phases of eigenstates of observable according to the evolution of the Heisenberg equation \eqref{eq:HeisenbergEquTime} or equivalent \eqref{eq:SchrodingerEquTimeEigenstate}. This is in spirit the same as the definition of the Aharonov-Anandan (AA) geometric phase \cite{AA1987}.

\item When some eigenvalues of the initial observable $X_0$ are degenerate as eigenstates, this would lead to the notion of non-Abelian observable-geometric phase as similar to the usual non-Abelian geometric phase (cf. \cite{Anandan1988,STAHJS2012}). This could be done as defining the geometric phases for mixed states as in \cite{SPEAEOV2000,STBCD2003}). We will discuss it elsewhere.

\end{enumerate}
\end{remark}

Although the observable-geometric phases in \eqref{eq:q-GeoPhase} are defined as similar to the AA geometric phase, the geometric interpretation is completely different. Indeed, $C_W: [0, T] \ni t \mapsto O(t) = \{ |\psi_n (t) \rangle \langle \psi_n (t)|: 1 \le n \le d \}$ defines a closed curve in the observable space $\mathcal{W} (\mathbb{H})$ which is a set of all complete orthogonal decompositions (see Section \ref{Appendix} for the notations). The equation \eqref{eq:ParallelCondVect} can be expressed as a condition (see \eqref{eq:CanonicalParallelCond} below) for quantum parallel transportation along $C_W.$ Therefore, the geometric interpretation of $\beta_n$'s defined in \eqref{eq:q-GeoPhase} is related with the geometric properties of $C_W$ determined by a quantum geometric structure over the observable space $\mathcal{W} (\mathbb{H}),$ distinct from that the AA geometric phase is associated with classical (Riemannian) geometry of the state space. We will give the details of this geometric interpretation in Section \ref{GeoInter}.

If the system evolves adiabatically (cf. \cite[Chapter 2]{BMKNZ2003}), $h(t)$ varies slowly with $h(t) |n(t)\rangle= E_n (t) |n(t)\rangle,$ for a complete set $\{|n(t)\rangle \},$ such that the state remains an eigenstate of $h(t)$ at all time $t$ with the same energy quantum number $n,$ namely the evolution operator
$$
U(t,0) \backsimeq \sum_n e^{\mathrm{i}\int^t_0 [ \langle n(s) | \mathrm{i} \frac{d}{d s} | n(s) \rangle - E_n (s)] ds} |n(t)\rangle \langle n(0)|
$$
to a good approximation, see \cite[(5)-(6)]{AA1987} or \cite[(2.37)-(2.39)]{BMKNZ2003} for the details. For an adiabatic cyclic evolution with period $T,$ i.e., $h(0) = h(T),$ such that all the eigenvalues of $h(0)$ are non-degenerate, the observable-geometric phases of the cyclic evolution of observable $X(t) = U(0,t)h(0)U(t,0)$ are
\begin{equation}\label{eq:q-GeoPhaseAdiaEvo}
\beta_n = - \int^T_0 [ \langle n(t) | \mathrm{i} \frac{d}{d t} | n(t) \rangle - E_n (t)] dt - \int^T_0 \langle n (0) | h(t) | n (0)\rangle d t
\end{equation}
because
$$
U(0,T)|n(0)\rangle = e^{-\mathrm{i}\int^T_0 [ \langle n(s) | \mathrm{i} \frac{d}{d s} | n(s) \rangle - E_n (s)] ds} |n(0)\rangle.
$$
In this case, in general, the observable-geometric phase $\beta_n$'s do not coincide with Berry's phase
\begin{equation}\label{eq:BerryGeoPhase}
\gamma_n = \int^T_0 \langle n(t) | \mathrm{i} \frac{d}{d t} | n(t) \rangle d t
\end{equation}
as define in \cite{Berry1984, Simon1983} (cf. \cite{AA1987}). But the observable-geometric phases $\beta_n$'s, defined by \eqref{eq:q-GeoPhase}, do not depend on any adiabatic evolution restriction, and can be defined for any period evolution of observable in a general quantum system, not just Hamiltonian $h(t)$ in an adiabatic cyclic system.

\section{Examples}\label{Ex}

For illustrating the observable-geometric phases, we first consider a qubit case, namely the Hilbert space $\mathbb{H} = \mathbb{C}^2$ and the observable space is $\mathcal{W} (\mathbb{C}^2).$ We then show how to make use of the Lewis-Resenfeld invariant theory \cite{LR1969} computing the observable-geometric phase of a general quantum system. In what follows, we will apply \eqref{eq:q-GeoPhase} to qubit systems, respectively subject to an orientated magnetic field and a rotating magnetic field.

Recall that the Pauli matrixes $\vec{\sigma} = (\sigma_x, \sigma_y, \sigma_z)$ are defined by
$$
\sigma_x = \left ( \begin{matrix} 0 & 1 \\
1 & 0
\end{matrix}\right ),\; \sigma_y = \left ( \begin{matrix} 0 & - \mathrm{i} \\
\mathrm{i} & 0
\end{matrix}\right ),\;\sigma_z = \left ( \begin{matrix} 1 & 0 \\
0 & -1
\end{matrix}\right ).
$$

\begin{example}\label{ex:qubit}\rm (An orientated magnetic field)\;
Consider that a spin-$\frac{1}{2}$ particle with a magnetic moment is in a homogeneous magnetic field $\vec{B} = (0,0,B)$ along the $z$ axis, whose Hamiltonian is $H = - \mu B \sigma_z/2$ where $\mu$ is the Bohr magneton. Given a spin observable $X_0$ with two non-degenerate eigenstates
\begin{equation}\label{eq:InitalObsQubit}
\psi_1 = \left ( \begin{matrix} \cos \frac{\phi}{2} \\
\sin \frac{\phi}{2}
\end{matrix}\right ),\;  \psi_2 = \left ( \begin{matrix} -\sin \frac{\phi}{2} \\
\cos \frac{\phi}{2}
\end{matrix}\right )
\end{equation}
in $\mathbb{C}^2,$ $X(t) = U(0,t) X_0 U(t,0)$ satisfies Eq.\eqref{eq:HeisenbergEquTime} with $\tilde{h} (t) = h (t)= - \mu B \sigma_z/2$ and $U(t,0) = e^{\mathrm{i} \mu B t \sigma_z/2}$ so that $\psi_n (t) = U(0,t) \psi_n$ ($n=1,2$), where
$$
\psi_1 (t) = e^{- \mathrm{i} \mu B t \sigma_z/2} \psi_1 = \left ( \begin{matrix} e^{- \mathrm{i} \mu B t/2} \cos \frac{\phi}{2} \\
e^{ \mathrm{i}\mu B t/2} \sin \frac{\phi}{2}
\end{matrix}\right ),\;  \psi_2 (t) = e^{- \mathrm{i}\mu B t \sigma_z/2} \psi_2 = \left ( \begin{matrix} - e^{- \mathrm{i}\mu B t/2} \sin \frac{\phi}{2} \\
e^{\mathrm{i}\mu B t/2} \cos \frac{\phi}{2}
\end{matrix}\right ),
$$
satisfies the skew Schr\"{o}dinger equation \eqref{eq:SchrodingerEquTimeEigenstate}, namely
$$
\mathrm{i} \frac{d \psi_n (t) }{d t} =\mu B \frac{\sigma_z}{2} \psi_n (t),\quad n=1,2.
$$
The evolution $X(t)$ is periodic with period $T = \frac{2\pi}{\mu B},$ namely $X(T) = X_0$ and precisely
$$
\psi_1 (T) = e^{\mathrm{i} \pi} \psi_1,\quad \psi_2 (T) = e^{\mathrm{i} \pi} \psi_2.
$$
By \eqref{eq:q-GeoPhase} we have
$$
\beta_1 = \pi (1 + \cos \phi ),\quad \beta_2 = \pi (1 - \cos \phi ),
$$
which are the geometric invariants of the curve $C_W:$
$$
(0, \frac{\pi}{\mu B}) \ni t \longmapsto \bigg \{ \left ( \begin{matrix}  \cos^2 \frac{\phi}{2} & \frac{1}{2} e^{- 2 \mathrm{i}\mu B t} \sin \phi \\
\frac{1}{2} e^{2 \mathrm{i}\mu B t} \sin \phi  & \sin^2 \frac{\phi}{2}
\end{matrix}\right ),\; \left ( \begin{matrix}  \sin^2 \frac{\phi}{2} & - \frac{1}{2} e^{- 2 \mathrm{i}\mu B t} \sin \phi \\
- \frac{1}{2} e^{2 \mathrm{i}\mu B t} \sin \phi  & \cos^2 \frac{\phi}{2}
\end{matrix}\right ) \bigg \}
$$
in $\mathcal{W} (\mathbb{C}^2).$

Indeed, taking a fixed point $O_0 = \{ |e_1\rangle \langle e_1|, |e_2\rangle \langle e_2| \} \in \mathcal{W} (\mathbb{C}^2),$ and letting
$$
\tilde{U} (t) = e^{\mathrm{i} t \frac{\mu B}{2}\cos \phi} | \psi_1(t) \rangle \langle e_1| + e^{- \mathrm{i} t \frac{\mu B}{2}\cos \phi} | \psi_2(t) \rangle \langle e_2|
$$
for $t \in (0, \frac{2\pi}{\mu B}),$ we have that $\tilde{C}_P: (0, \frac{\pi}{\mu B}) \ni t \longmapsto \tilde{U} (t)$ is the horizontal $O_0$-lift of $C_W$ with respect to $\check{\Omega}$ in the principal bundle $\xi_{O_0}$ (the definitions of the canonical connection $\check{\Omega}$ and the bundle $\xi_{O_0}$ refer to Section \ref{Appendix}) so that
$$
\tilde{U} (T) = e^{-\mathrm{i} \beta} |\psi_1 \rangle \langle e_1| + e^{\mathrm{i} \beta} |\psi_2 \rangle \langle e_2|
$$
with $\beta = \beta_2 =\pi (1 - \cos \phi )$ (because $\beta_1 = 2 \pi - \beta_2$) is the holonomy element associated with the connection $\check{\Omega}, C_W,$ and $U_0 = \sum^2_{n = 1}| \psi_n \rangle \langle e_n |$ in $\xi_{O_0}.$
\end{example}

\begin{example}\label{ex:qubit-rotatingfield}\rm (A rotating magnetic field)\;
Consider a spin-$\frac{1}{2}$ particle subject to a rotating background magnetic field
$$
\mathbf{B} (t) = (B_0 \cos \omega t, B_0 \sin \omega t, B_1)
$$
with a constant angular velocity $\omega,$ where $B_0$ and $B_1$ are constants, whose Hamiltonian is given by
$$
h (t)= - \frac{\mu}{2}\mathbf{B}(t) \cdot \vec{\sigma} = - \frac{1}{2} (\omega_0 \sigma_x \cos \omega t + \omega_0 \sigma_y \sin \omega t + \omega_1 \sigma_z)
$$
with $\omega_i = \mu B_i$ ($i=0,1$). Since
$$
\sigma_x \cos \omega t + \sigma_y \sin \omega t = e^{- \mathrm{i} \omega t \sigma_z/2} \sigma_x e^{\mathrm{i} \omega t \sigma_z/2 },
$$
the Hamiltonian takes the form $h(t) = e^{- \mathrm{i} \omega t \frac{\sigma_z}{2}} h_0 e^{\mathrm{i} \omega t \frac{\sigma_z}{2}}$ with
$$
h_0 = - \frac{1}{2} \omega_0 \sigma_x - \frac{1}{2}\omega_1 \sigma_z.
$$
If $\phi (t)$ satisfies the time-dependent Schr\"{o}dinger equation \eqref{eq:SchrodingerEquTime}, then $\psi (t) = e^{\mathrm{i} \omega t \frac{\sigma_z}{2}} \phi (t)$ satisfies the Schr\"{o}dinger equation
$$
\mathrm{i} \frac{d \psi (t)}{d t} = (h_0 - \frac{1}{2}\omega \sigma_z ) \psi (t).
$$
Because
$$
H = h_0 - \frac{1}{2}\omega \sigma_z = - \frac{1}{2} \omega_0 \sigma_x - \frac{1}{2}(\omega_1 + \omega) \sigma_z
$$
is time-independent, we have $\psi (t) = e^{-\mathrm{i} tH} \psi (0)$ and thus
$$
\phi (t) = e^{-\mathrm{i} \omega t \frac{\sigma_z}{2}} e^{-\mathrm{i} tH} \psi (0)
$$
so that $U(t,0) = e^{-\mathrm{i} \omega t \frac{\sigma_z}{2}} e^{-\mathrm{i} tH}$ is the propagator associated with $h(t).$

Having obtained the exact expression of the propagator $U(t,0)$ we can next identify the cyclic evolution of observable, namely the solutions which satisfy \eqref{eq:HeisenbergEquTime} with
$$
\tilde{h} (t) = U(0,t) h(t) U(t, 0) = e^{\mathrm{i} tH} h_0 e^{-\mathrm{i} tH}.
$$
Given a spin observable $X_0$ with non-degenerate eigenstates $\psi_1$ and $\psi_2,$ $X(t) = U(0,t) X_0 U(t,0)$ is by definition cyclic with period $T$ if and only if $\{\psi_1, \psi_2\}$ is a complete set of eigenstates of $U(0, T).$ The corresponding eigenvalues $e^{\mathrm{i} \theta_n}$ are the total phase factors, namely $U(0, T)\psi_n = e^{\mathrm{i} \theta_n}\psi_n.$ The corresponding observable-geometric phases are
\begin{equation}\label{eq:GeophaseQubit}
\beta_n = \theta_n + \frac{\omega_0}{2\omega} \langle \psi_n | \sigma_x | \psi_n \rangle \sin \omega T + \frac{\omega_0}{2\omega} \langle \psi_n | \sigma_y | \psi_n \rangle (1-\cos \omega T) + \frac{\omega_1}{2}\langle \psi_n | \sigma_z | \psi_n \rangle T
\end{equation}
for every $n =1, 2.$ As shown in Theorem \ref{thm:q-GeoPhase}, $\beta_1, \beta_2$ are two geometric invariants of the curve
$$
C_W: [0,T] \ni t \mapsto \{U(0,t) |\psi_n \rangle\langle \psi_n| U(t,0): n =1, 2\}
$$
in $\mathcal{W} (\mathbb{C}^2),$ and
$$
\tilde{U} (T) = e^{\mathrm{i} \beta_1} |\psi_1 \rangle \langle e_1| + e^{\mathrm{i} \beta_2} |\psi_2 \rangle \langle e_2|
$$
is the holonomy element associated with the connection $\check{\Omega}, C_W,$ and $U_0 = \sum^2_{n = 1}| \psi_n \rangle \langle e_n |$ in $\xi_{O_0},$ given an arbitrary point $O_0 = \{ |e_1\rangle \langle e_1|, |e_2\rangle \langle e_2| \}$ in $\mathcal{W} (\mathbb{C}^2).$

A special case is the period $T= 2 \pi/\omega,$ namely the same period as the Hamiltonian $h(t).$ In this case, cyclic evolutions are obtained by finding the simultaneous eigenvectors of the operator $U(0,2 \pi/\omega) = e^{\mathrm{i} \frac{2\pi}{\omega} H} e^{\mathrm{i} \pi \sigma_z} = -e^{\mathrm{i} \frac{2\pi}{\omega} H}.$ These eigenvectors are consequently also ones of the operator $H = - \frac{1}{2} \omega_0 \sigma_x - \frac{1}{2}(\omega_1 + \omega) \sigma_z,$ a complete set of whose eigenvectors is
$$
\psi_+ = \left ( \begin{matrix} \cos \frac{\phi}{2} \\
\sin \frac{\phi}{2}
\end{matrix}\right ),\;  \psi_- = \left ( \begin{matrix} -\sin \frac{\phi}{2} \\
\cos \frac{\phi}{2}
\end{matrix}\right ),
$$
namely $H \psi_\pm = \pm \frac{1}{2}\sqrt{\omega^2_0 + (\omega_1 + \omega)^2} \psi_\pm,$ where
$$
\phi = 2 \arctan \frac{\omega_0}{\omega_1 + \omega + \sqrt{\omega^2_0 + (\omega_1 + \omega)^2}}.
$$
Then, for the cyclic evolution $X(t)=U(0,t) X_0 U(t,0)$ starting from the initial observable $X_0$ with two non-degenerate eigenstates $\psi_\pm,$ the total phase factors are
$$
\theta_\pm = \pi \pm \frac{\pi}{\omega} \sqrt{\omega^2_0 + (\omega_1 + \omega)^2}
$$
because $U(0, 2 \pi/\omega) \psi_\pm = e^{\mathrm{i} \theta_\pm}\psi_\pm.$ Thus, by \eqref{eq:GeophaseQubit} the corresponding observable-geometric phases are
$$
\beta_\pm = \pi \pm \frac{\pi}{\omega} \Big ( \sqrt{\omega^2_0 + (\omega_1 + \omega)^2} + \omega_1 \cos \phi \Big ),
$$
which are two geometric invariants of the curve
$$
C_W: [0,T] \ni t \mapsto \{U(0,t) |\psi_+ \rangle\langle \psi_+| U(t,0), U(0,t) |\psi_- \rangle\langle \psi_-| U(t,0)\}
$$
in $\mathcal{W} (\mathbb{C}^2).$ Moreover,
\begin{equation}\label{eq:GeoPhaseQubitUnitaryOper}
\tilde{U} (T) = e^{\mathrm{i} \beta} |\psi_+ \rangle \langle e_1| + e^{-\mathrm{i} \beta} |\psi_- \rangle \langle e_2|
\end{equation}
with $\beta= \beta_+$ (because $\beta_- = 2\pi- \beta_+$) is the holonomy element associated with the connection $\check{\Omega}, C_W,$ and $U_0 = |\psi_+ \rangle \langle e_1| + |\psi_- \rangle \langle e_2|$ in $\xi_{O_0},$ given an arbitrary point $O_0 = \{ |e_1\rangle \langle e_1|, |e_2\rangle \langle e_2| \}$ in $\mathcal{W} (\mathbb{C}^2).$
\end{example}

Generally speaking, it is a difficult task for getting a explicit expression of the evolution operator of a quantum system with a time-dependent Hamiltonian. To this end, Lewis and Riesenfeld \cite{LR1969} developed a dynamical invariant theory for computing the evolution operator. Recall that a dynamical (time-dependent) observable $I(t)$ is said to be an invariant for the system with a time-dependent Hamiltonian $h(t),$ if it obeys the Liouville-von Neumann equation
$$
\mathrm{i} \frac{\partial I(t)}{\partial t} = [h(t), I(t)],
$$
namely its expectation value is a constant. It is shown in \cite{LR1969} that any eigenvalue of $I(t)$ is independent of time $t,$ and if $|\varphi(t) \rangle$ is a solution of Schr\"{o}dinger's equation \eqref{eq:SchrodingerEquTime}, so does $I(t)|\varphi (t)\rangle.$ Let $\{|\varphi_n (t) \rangle \}$ be a complete set of eigenstates of $I(t).$ Then the evolution operator
\begin{equation}\label{eq:EvoluOperInOb}
U(t,0) = \sum_n e^{\mathrm{i} \alpha_n(t)} |\varphi_n (t) \rangle \langle \varphi_n (0)|,
\end{equation}
where
\begin{equation}\label{eq:Lewis-Riesenfeld phase}
\alpha_n (t) = \int^t_0 \big [\langle \varphi_n (s)| \mathrm{i} \frac{d}{d s}|\varphi_n (s)\rangle - \langle \varphi_n (s)| h(s)|\varphi_n (s)\rangle \big ] d s
\end{equation}
which are called the Lewis-Riesenfeld phase (cf. \cite{GXQ1991,MKN1994,GWN2014}).

If the invariant $I(t)$ is $T$-periodic and the eigenvalues of $I(0)$ are non-degenerate, then $I(0)$ is a cyclic observable with period $T,$ such that
$$
U(0, T) |\varphi_n (0) \rangle = e^{-\mathrm{i} \alpha_n(T)} \langle \varphi_n (t)|\varphi_n (0) \rangle |\varphi_n (0) \rangle.
$$
Let us take a continuous, closed representative of this curve $|\varphi_n (t) \rangle,$ i.e., $|\varphi_n (0) \rangle = |\varphi_n (T) \rangle.$ Then, the observable-geometric phases of the cyclic evolution $X(t) = U(0,t) I(0) U(t,0)$ are
\begin{equation}\label{eq:q-GeoPhaLRphase}
\beta_n = - \alpha_n(T) - \int^T_0 \langle \varphi_n (0) |h(t)|\varphi_n (0) \rangle d t.
\end{equation}
In this case, in general, the observable-geometric phase $\beta_n$'s do not coincide with the AA phase
\begin{equation}\label{eq:BerryGeoPhase}
\gamma_n = \int^T_0 \langle \varphi_n (t) | \mathrm{i} \frac{d}{d t} | \varphi_n(t) \rangle d t
\end{equation}
as define in \cite{AA1987} (cf. \cite{GWN2014}).
In what follows, we apply the formula \eqref{eq:q-GeoPhaLRphase} to the qubit system.

\begin{example}\rm (The Lewis-Riesenfeld phase for a qubit)\;
Consider a qubit system with time-dependent Hamiltonian given by
$$
h (t)= \frac{1}{2} (\Omega_0 \sigma_x \cos \omega t + \Omega_0 \sigma_y \sin \omega t + \Omega_1 \sigma_z),
$$
where $\Omega_0, \Omega_1,$ and $\omega$ are all constants, see Example \ref{ex:qubit-rotatingfield} for the details. The corresponding dynamical invariant for this system is
\begin{equation}\label{eq:LRinvariantQubit}
I(t) = \Omega_0 \sigma_x \cos \omega t + \Omega_0 \sigma_y \sin \omega t + (\Omega_1 - \omega) \sigma_z
\end{equation}
(cf. \cite{GWN2014}), whose eigenvalues
$$
\lambda_\pm = \pm \sqrt{\Omega^2_0 + (\Omega_1 -\omega)^2}
$$
with the corresponding eigenstates
$$
\varphi_\pm (t) = \left ( \begin{matrix} \frac{e^{-\mathrm{i} \omega t} \Delta_\pm}{\sqrt{1 + \Delta^2_\pm}}\\
\frac{1}{\sqrt{1 + \Delta^2_\pm}}
\end{matrix}\right )
$$
where $\Delta_\pm = \frac{1}{\Omega_0} (\lambda_\pm + \Omega_1 - \omega).$ By direct computation, the corresponding Lewis-Riesenfeld phases of these vectors are
\begin{equation}\label{eq:LRPhaseQubit}
\alpha_\pm (t) = \frac{1}{2} (\omega \mp \lambda) t
\end{equation}
so that
$$
U(t,0) = e^{\frac{\mathrm{i}}{2} (\omega - \lambda) t} |\varphi_+ (t) \rangle \langle \varphi_+ (0)| + e^{\frac{\mathrm{i}}{2} (\omega + \lambda) t} |\varphi_- (t) \rangle \langle \varphi_- (0)|.
$$
Clearly, $I(t)$ is cyclic in time with period $T = \frac{2 \pi}{\omega},$ so does $h(t).$ By \eqref{eq:q-GeoPhaLRphase}, the observable-geometric phases of the cyclic evolution $X(t) = U(0,t) I(0) U(t,0)$ are
\begin{equation}\label{eq:q-GeoPhaLRphaseQubit}
\beta_\pm = \pi \pm \frac{\pi \lambda}{\omega} - \frac{\Omega_1 (\Delta^2_\pm -1)}{2 (1 + \Delta^2_\pm)}.
\end{equation}
\end{example}

\section{Geometric quantum computation}\label{GQC}

In this section, we shall show that the observable-geometric phases can be used to realize a set of universal quantum gates. As is well-known, for a universal set of quantum gates, we need to achieve two kinds of noncommutative $1$-qubit gates and one nontrivial $2$-qubit gate (cf. \cite{DBE1995, Lloyd1995}). Thus, we need to realize two noncommutative $1$-qubit gates and a nontrivial $2$-qubit gate as the geometric evolution operators on loops in the observable spaces of a qubit and $2$-qubits, respectively.

\begin{proposition}\label{prop:1-qubitGeoGate}\rm
The two $1$-qubit gates $e^{\mathrm{i}\pi |1\rangle \langle 1|}$ and $e^{\mathrm{i} \sigma_x}$ can be respectively realized as geometric quantum gates associated with the geometry of the observable space $\mathcal{W} (\mathbb{C}^2)$ of a qubit, which are two well-known gates constituting a universal set of quantum gates for a single-qubit system.
\end{proposition}

\begin{proof}
As in Example \ref{ex:qubit-rotatingfield}, a general time-dependent Hamiltonian for a qubit has only one term $h (t) = - \mu {\bf B} (t) \cdot \vec{\sigma}/2,$ where ${\bf B} (t) = (B_0 \cos \omega t, B_0 \sin \omega t, B_1)$ denotes the total magnetic field felt by the qubit, $B_0, B_1$ and $\omega$ are constants, and $\mu$ is the Bohr magneton. Given a spin observable $X_0$ with two non-degenerate eigenstates
$$
\psi_+ = \cos \frac{\phi}{2} | 0 \rangle + \sin \frac{\phi}{2} | 1 \rangle,\; \psi_- = - \sin \frac{\phi}{2} | 0 \rangle + \cos \frac{\phi}{2} | 1 \rangle,
$$
where $|0\rangle$ and $|1\rangle$ constitute the computational basis for the qubit, the observable evolution $t \longmapsto U(0,t) X_0 U(t,0)$ is periodic with period $T = 2\pi/\omega$ provided
$$
\phi = 2 \arctan \frac{\omega_0}{\omega_1 + \omega + \sqrt{\omega^2_0 + (\omega_1 + \omega)^2}}
$$
with $\omega_i = \mu B_i$ ($i=0,1$), where $U(t,0) = e^{-\mathrm{i} \omega t \frac{\sigma_z}{2}} e^{-\mathrm{i} tH}$ with $H= -[\omega_0 \sigma_x + (\omega_1 + \omega)\sigma_z]/2.$ This cyclic process determines the evolution operator $U (T, 0) = e^{-\mathrm{i} \theta} |\psi_+\rangle \langle \psi_+| + e^{\mathrm{i} \theta} |\psi_-\rangle \langle \psi_-|$ with $\theta = \pi + \frac{\pi}{\omega} \sqrt{\omega^2_0 + (\omega_1 + \omega)^2}.$ If we can remove the dynamical phase from the total phase $\theta,$ then the evolution operator becomes a geometric quantum gate. This can be done by using a two-loop method as in the case of the AA geometric phases (cf. \cite{ZW2002, ZW2003}).

To this end, we first allow the time-dependent Hamiltonian $h(t)$ to go through cyclic evolution with period $T= 2\pi/\omega.$ Precisely, we consider the process where a spin observable $X_0$ with two non-degenerate eigenstates $\psi_\pm$ can evolve cyclically. We first decide the cyclic evolution observable $X_0$ by choosing $\phi$ as above. The phases of eigenstates $\psi_\pm$ acquired in this way would contain both a geometric and a dynamical component as described in \eqref{eq:q-GeoPhase}. In order to remove the dynamical phase accumulated in the above process, we allow $X_0$ to evolve along the time-reversal path of the first-period loop during the second period. This process can be realized by reversing the effective magnetic field with $\mathbf{B} (t+T)= - \mathbf{B} (T-t)$ on the same loop of the first-period $[0, T),$ and thus $h(t+T) = - h(T-t)$ for $t \in [0, T).$ As a result, the total phases accumulated in the two periods will be just the observable-geometric phases, because the dynamical phase $\gamma^{(d)}_n = \int^T_0 \langle \psi_n | h(t) | \psi_n \rangle d t$ appearing in \eqref{eq:q-GeoPhase} will be canceled.

Indeed, the dynamical phase $\gamma^{(2,d)}_n$ for the second period is equal to that ($\gamma^{(1,d)}_n$) for the first period with the opposite sign, namely $\gamma^{(2,d)}_n = -\gamma^{(1,d)}_n,$ since
\begin{equation*}
\begin{split}
\gamma^{(2,d)}_n & = \int^{2T}_T \langle \psi_n | h(t) | \psi_n \rangle d t = \int^T_0 \langle \psi_n | h(t+T) | \psi_n \rangle d t\\
& = - \int^T_0 \langle \psi_n | h(T-t) | \psi_n \rangle d t = - \int^T_0 \langle \psi_n | h(t) | \psi_n \rangle d t = - \gamma^{(1,d)}_n.
\end{split}
\end{equation*}
On the other hand, the propagator for the second period is $U(t, 0) = e^{\mathrm{i} \omega (t-2T)\sigma_z/2} e^{\mathrm{i}t H}$ for $t \in (T, 2T].$ Thus, the observable-geometric phases are
$$
\beta_\pm = \mp \frac{\pi}{\omega} \sqrt{\omega^2_0 + (\omega_1 + \omega)^2},
$$
namely $U(2T,0) = e^{\mathrm{i} \beta} |\psi_+\rangle \langle \psi_+| + e^{-\mathrm{i} \beta} |\psi_-\rangle \langle \psi_-|$ with $\beta = \beta_-.$

Therefore, by removing the dynamical phases, we obtain the evolution operator $U(2T,0) = U_{\phi, \beta}$ with the matrix representation
\begin{equation}\label{eq:ParallelUnitaryOper1-bitgate}
U_{\phi, \beta} = \left ( \begin{matrix} e^{\mathrm{i} \beta} \cos^2 \frac{\phi}{2} + e^{- \mathrm{i} \beta} \sin^2 \frac{\phi}{2} & \mathrm{i} \sin \phi \sin \beta\\
\mathrm{i} \sin \phi \sin \beta & e^{\mathrm{i} \beta} \sin^2 \frac{\phi}{2} + e^{-\mathrm{i} \beta} \cos^2 \frac{\phi}{2}
\end{matrix}\right ).
\end{equation}
As noted above, $U_{\phi, \beta}$ depends only on the geometric property of the curve $C_W$ and $\phi,$ and thus is a geometric quantum gate. We write an input state as $| \psi_\mathrm{in} \rangle = \alpha_+ | 0 \rangle + \alpha_- | 1 \rangle$ with $\alpha_\pm = \langle \psi_\pm (0) | \psi_\mathrm{in} \rangle.$ Then $U_{\phi, \beta}$ is a single-qubit gate such that $| \psi_\mathrm{out} \rangle = U_{\phi, \beta} | \psi_\mathrm{in} \rangle.$

By computing, two operations $U_{\phi, \beta}$ and $U_{\phi', \beta'}$ do not commute if and only if $\phi \not= \phi' + k \pi$ and $\beta, \beta' \not= k \pi,$ where $k$ is an integer. By choosing $B_0, B_1$ and $\omega$ so that $\phi =0, \phi' = \frac{\pi}{2},$ and $\beta = \beta' = \frac{\pi}{2},$ we get the two single-qubit gates
$$
U_{0, \pi/2} = \mathrm{i} e^{\mathrm{i} \pi |1\rangle \langle 1|},\quad U_{\frac{\pi}{2}, \frac{\pi}{2}} = e^{\mathrm{i} \sigma_x},
$$
This completes the proof.
\end{proof}

Next, we turn to a nontrivial $2$-qubit gate, namely the controlled-NOT gate (c-NOT), which is defined as $|k \rangle |m\rangle \mapsto |k\rangle |k \oplus m \rangle$ for $k,m =0,1,$ where $\oplus$ denotes the addition modulo $2.$

\begin{proposition}\label{prop:2-qubitGeoGate}\rm
The controlled-NOT gate can be realized as a geometric quantum gate associated with the geometry of the observable space $\mathcal{W} (\mathbb{C}^2 \otimes \mathbb{C}^2)$ of the two-qubits system.
\end{proposition}

\begin{proof}
Consider a $2$-qubit quantum system with a time-dependent Hamiltonian $h(t),$ where the first qubit is the control qubit while the second the target one. Suppose $h(t)$ be of the form
$$
h(t) = \left ( \begin{matrix} h_0 (t) & 0\\
0 & h_1 (t)
\end{matrix}\right ).
$$
This is the case when the control qubit is far away from the resonance condition for the operation of the target qubit. Given an initial $2$-qubit observable $X_0$ with four non-degenerate eigenstates:
$$\psi^{(0)}_+ = |0 \rangle \otimes \big ( \cos \frac{\phi^{(0)}}{2} |0 \rangle + \sin \frac{\phi^{(0)}}{2} |1 \rangle \big ), \psi^{(0)}_- = |0 \rangle \otimes \big ( - \sin \frac{\phi^{(0)}}{2} |0 \rangle + \cos \frac{\phi^{(0)}}{2} |1 \rangle \big ),$$
and
$$\psi^{(1)}_+ = |1 \rangle \otimes \big ( \cos \frac{\phi^{(1)}}{2} |0 \rangle + \sin \frac{\phi^{(1)}}{2} |1 \rangle \big ), \psi^{(1)}_- = |1 \rangle \otimes \big ( - \sin \frac{\phi^{(1)}}{2} |0 \rangle + \cos \frac{\phi^{(1)}}{2} |1 \rangle \big ),$$
where $\{ |00 \rangle, |01 \rangle, |10 \rangle, |11 \rangle \}$ is the computational basis, the evolution $t \longmapsto U(0,t) X_0 U(t,0)$ is cyclic with period $T$ under the suitable choices of $\phi^{(0)}$ and $\phi^{(1)}$ so that $\psi^{(k)}_\pm$ are eigenstates of $U(T,0)$ for $k=0,1.$ In this case, the evolution operator to describe the two-qubit gate is given by
\begin{equation}\label{eq:ParallelUnitaryOper2-bitgate}
U (\phi^{(0)}, \beta^{(0)}; \phi^{(1)}, \beta^{(1)}) = \left ( \begin{matrix} U_{\phi^{(0)}, \beta^{(0)}} & 0\\
0 & U_{\phi^{(1)}, \beta^{(1)}}
\end{matrix}\right ),
\end{equation}
where $\beta^{(k)}$ ($\phi^{(k)}$) denotes the geometric phase (the cyclic initial observable) of the target bit when the state of the control qubit corresponds to $|k \rangle$ with $k=0,1,$ respectively.

To achieve the controlled-NOT gate, we take $h_0 =0$ and $h_1 (t) =- \mu {\bf B} (t) \cdot \vec{\sigma}/2$ with ${\bf B} (t) = (B_0 \cos \omega t, B_0 \sin \omega t, B_1)$ as in Proposition \ref{prop:1-qubitGeoGate}. Then the gate in this case with $T=2\pi/\omega$ is given by
$$
U (\frac{\pi}{2}, 0; \phi, \beta) = \left ( \begin{matrix} I & 0\\
0 & U_{\phi, \beta}
\end{matrix}\right ),
$$
where $\beta$ is the total phase accumulated in the evolution when the controlled qubit is in the state $|1\rangle,$ and $\phi$ is chose as in Proposition \ref{prop:1-qubitGeoGate}. By a two-loop method as above, we can remove the dynamical phase and obtain the geometric quantum gate
$$
U \big (\frac{\pi}{2}, 0; \frac{\pi}{2}, \frac{\pi}{2} \big ) =  \left ( \begin{matrix} I & 0\\
0 & \mathrm{i} \sigma_x
\end{matrix}\right ).
$$
This gate is equivalent to the controlled-NOT gate, up to an overall phase factor $\mathrm{i}$ for the target qubit. Therefore, we get a nontrivial $2$-qubit gate in terms of the observable-geometric phases.
\end{proof}

\begin{remark}\rm
\begin{enumerate}[\rm (1)]

\item Universal quantum computing has been achieved by Zhu and Wang \cite{ZW2002,ZW2003} using the gates associated with the AA geometric phase, as well by Sj\"{o}qvist {\it et al} \cite{STAHJS2012} using the non-Abelian geometric phase (cf. \cite{Anandan1988}). As shown in \cite{SMC2016}, the non-Abelian geometric gates of Sj\"{o}qvist {\it et al} can be interpreted as the Abelian geometric gates of Zhu and Wang.

\item The quantum gate based on the observable-geometric phases can be also interpreted as the Abelian geometric gates of Zhu and Wang; indeed the two schemes lead to the same gates as shown in Propositions \ref{prop:1-qubitGeoGate} and \ref{prop:2-qubitGeoGate}. However, there are two differences between them: 1)\; the ZW gates are related with the geometry of the state space of the system, while ours are connected with the geometry of the observable space; 2)\; the gates in the ZW scheme are all based at the simultaneous evolutions of two orthogonal basic vectors, while in our setting all gates are based at the evolution of a single observable.

\end{enumerate}
\end{remark}

\section{Conclusions}\label{Conclusion}

Geometric phases provide a new way of looking at quantum mechanics. The usual theory of the geometric phase (cf. \cite{AA1987, BMKNZ2003}) is based on the Schr\"{o}dinger picture, that is, the geometric phase is defined for the state. Here, we define the observable-geometric phases in the Heisenberg picture. This provide a new way of studying the geometry for the quantum system from the viewpoint of the observable. In particular, the observable-geometric phases can be used to realize a universal set of quantum gates in quantum computation. Therefore, it may not be unreasonable to hope that this new insight may have heuristic value.

\section{Appendix: Quantum geometry over the observable space}\label{Appendix}

\subsection{Observable space}\label{ObservableSpace}

Recall that a complete orthonormal decomposition in $\mathbb{H}$ is a set of projections of rank one $\{ | n \rangle \langle n|: n \ge 1 \}$ satisfying
\begin{equation}\label{eq:OrthDecomp}
\sum_{n \ge 1} | n \rangle \langle n| = I,\quad \langle n|m \rangle = \delta_{n m}.
\end{equation}
We denote by $\mathcal{W} (\mathbb{H})$ the set of all complete orthonormal decompositions in $\mathbb{H}.$ Note that a complete orthonormal decomposition $O= \{ | n \rangle \langle n|: n \ge 1 \}$ determines uniquely a basis $\{| n \rangle \}_{n \ge 1}$ up to phases for basic vectors. Conversely, a basis uniquely defines a complete orthonormal decomposition in $\mathbb{H}.$ Therefore, a complete orthonormal decomposition can be regarded as a complete measurement. Since every observable has a spectral decomposition corresponding to (at least) a complete orthonormal decomposition, the evolution of a quantum system by the Heisenberg equation
\begin{equation}\label{eq:HeisenbergEqu}
\mathrm{i} \frac{d X}{d t} = [X, H]
\end{equation}
for the observable $X,$ gives rise to a curve in  $\mathcal{W} (\mathbb{H}).$ This is the reason why $\mathcal{W} (\mathbb{H})$ can be regarded as the observable space, whose geometry and topology induce geometric and topological phases for quantum systems.

For $O, O'\in \mathcal{W} (\mathbb{H}),$ we define the distance $D_\mathcal{W} ( O, O')$ by
\begin{equation}\label{eq:HausdDist}
D_\mathcal{W} ( O, O') = \inf \{ \| I - U \|:\; U^{-1} O' U = O,\; U \in \mathcal{U} (\mathbb{H}) \}.
\end{equation}
It is easy to check that $\mathcal{W} (\mathbb{H})$ is a complete metric space together with the distance $D_\mathcal{W}.$ Thus, $\mathcal{W} (\mathbb{H})$ is a topological space as a metric space.

To obtain more information on $\mathcal{W} (\mathbb{H}),$ we define $\mathcal{X} (\mathbb{H})$ to be the set of all ordered sequences $(|n\rangle \langle n |)_{n \ge 1},$ where $\{|n\rangle \langle n |:\; n \ge 1\}$'s are all complete orthonormal decompositions in $\mathbb{H}.$ A distance on $\mathcal{X} (\mathbb{H})$ is defined as follows: For $(|n\rangle \langle n |)_{n \ge 1}, (|n'\rangle \langle n' |)_{n \ge 1} \in \mathcal{X} (\mathbb{H}),$
$$
D_\mathcal{X} ((|n\rangle \langle n |)_{n \ge 1}, (|n'\rangle \langle n' |)_{n \ge 1}) = \inf \{ \| I-U \|:\; U \in \mathcal{U} (\mathbb{H}),\; |n'\rangle = U |n \rangle, \forall n \}.
$$
Then it is easy to check that $D_\mathcal{X}$ is a distance such that $\mathcal{X} (\mathbb{H})$ is a complete metric space and
$$
\mathcal{W} (\mathbb{H}) \cong \frac{\mathcal{X} (\mathbb{H})}{\Pi (d)},
$$
where $\Pi (d)$ denotes the permutation group of $d$ objects, and the notation $\cong$ indicates the isometric isomorphism of two metric spaces. From this identification, we find that $\mathcal{W} (\mathbb{H})$ is topologically non-trivial as its fundamental group is isomorphic to $\Pi (d),$ since $\mathcal{X} (\mathbb{H})$ can be considered as a compact, simply connected manifold.

\begin{proposition}\label{prop:TopoSpaceQ-system}\rm
For an arbitrary fixed basis $(|e_n\rangle )_{n \ge 1}$ of $\mathbb{H},$
$$
\mathcal{W} (\mathbb{H}) \cong \{ \mathcal{G} (U):\; U \in \mathcal{U} (\mathbb{H})\}
$$
with
\begin{equation}\label{eq:FiberForm}
\mathcal{G} (U) = \Big \{ \sum_{n \ge 1} e^{ \mathrm{i} \theta_n} |\sigma (n) \rangle \langle e_n |:\; \forall \sigma \in \Pi (d), \forall \theta_n \in [0, 2 \pi ) \Big \},
\end{equation}
where $|n\rangle = U |e_n\rangle$ for any $n \ge 1,$ and the distance between two elements is defined by
$$
d (\mathcal{G} (U), \mathcal{G} (U')) = \inf \{ \| K - G \|: K \in \mathcal{G} (U), G \in \mathcal{G} (U') \}.
$$
\end{proposition}

\begin{proof}
At first, we prove that
$$
\mathcal{X} (\mathbb{H}) \cong \frac{\mathcal{U} (\mathbb{H})}{\mathcal{U} (1)^d},
$$
from which we conclude the result, since $\mathcal{W} (\mathbb{H}) \cong \frac{\mathcal{X} (\mathbb{H})}{\Pi (d)}.$

Indeed, for an arbitrary fixed basis $(|e_n\rangle )_{n \ge 1}$ of $\mathbb{H},$ we have that $\frac{\mathcal{U} (\mathbb{H})}{\mathcal{U} (1)^d} = \{ [U]:\; U \in \mathcal{U} (\mathbb{H}) \}$ with
$$
[U] = U \cdot \mathcal{U} (1)^d = \Big \{ \sum_{n \ge 1} e^{ \mathrm{i} \theta_n} |n\rangle \langle e_n |:\; \forall \theta_n \in [0, 2 \pi ) \Big \},
$$
where $|n\rangle = U |e_n\rangle$ for $n \ge 1.$ Define $T: \mathcal{X} (\mathbb{H}) \mapsto \frac{\mathcal{U} (\mathbb{H})}{\mathcal{U} (1)^\mathbb{N}}$ by
$$
T [(|n\rangle \langle n |)_{n \ge 1}] \longmapsto [U]
$$
for any $(|n\rangle \langle n |)_{n \ge 1} \in \mathcal{X} (\mathbb{H}),$ where $U$ is the unitary operator so that $|n\rangle = U |e_n\rangle$ for $n \ge 1.$ Then, $T$ is surjective and isometric, and so the required assertion follows. This completes the proof.
\end{proof}

For illustration, we consider the topology of $\mathcal{W} (\mathbb{H})$ in the qubit case of $\mathbb{H} = \mathbb{C}^2.$ Indeed, we have
$$
\mathcal{X} (\mathbb{C}^2) \cong \frac{\mathcal{U} (2)}{\mathcal{U}(1) \times \mathcal{U}(1)} \cong \frac{\mathcal{ S U} (2)}{\mathcal{U}(1)} \cong \mathbb{S}^2,
$$
and so
$$
\mathcal{W} (\mathbb{C}^2) \cong \frac{\mathbb{S}^2}{\mathbb{Z}_2},
$$
where we have used the fact $\Pi (2) = \mathbb{Z}_2.$ This has a simple geometrical interpretation, since every element in $\mathcal{X} (\mathbb{C}^2)$ has the form $X_{\vec{n}} = ( |-\vec{n} \rangle \langle -\vec{n}|, | \vec{n}\rangle \langle \vec{n}|)$ with $\vec{n} = (n_x, n_y, n_z) \in \mathbb{S}^2,$ where
$$
| \pm \vec{n}\rangle \langle \pm \vec{n}| = \frac{1}{2} (I \pm \vec{n} \cdot \vec{\sigma}).
$$
Although $X_{\vec{n}} \not= X_{-\vec{n}}$ in $\mathcal{X} (\mathbb{C}^2),$ they both correspond to the same element in $\mathcal{W} (\mathbb{C}^2).$ This implies that $\mathcal{W} (\mathbb{C}^2)$ may have non-trivial topology: There are exactly two topologically distinct classes of loops in $\mathcal{W} (\mathbb{C}^2),$ one corresponds to the trivial class $X_{\vec{n}} \longmapsto X_{\vec{n}}$ while the other to the nontrivial class $X_{\vec{n}} \longmapsto X_{-\vec{n}}.$ Then the first fundamental group $\pi_1 ( \mathcal{W} (\mathbb{C}^2) ) \cong \mathbb{Z}_2,$ and thus the topology of the observable space $\mathcal{W} (\mathbb{C}^2)$ for the qubit system is nontrivial.

\subsection{Fibre bundles over the observable space}\label{FibleBundleOb}

Following \cite{Isham1999}, a bundle is a triple $(E, \pi, B),$ where $E$ and $B$ are two Hausdorff topological spaces, and $\pi: E \mapsto B$ is a continuous map which is always assumed to be surjective. The space $E$ is called the total space, the space $B$ is called the base space, and the map $\pi$ is called the projection of the bundle. For each $b \in B,$ the set $\pi^{-1} (b)$ is called the fiber of the bundle over $b.$ Given a topological space $F,$ a bundle $(E, \pi, B)$ is called a fiber bundle with the fiber $F$ provided every fiber $\pi^{-1} (b)$ for $b \in B$ is homeomorphic to $F.$ For a topological group $G,$ a bundle $(E, \pi, B)$ is called a $G$-bundle, denoted by $(E, \pi, B, G),$ provided $G$ acts on $E$ from the right preserving the fibers of $E$ such that the map $f$ from the quotient space $E/G$ onto $B$ defined by $f (x G) = \pi (x)$ for $x G \in E/G$ is a homeomorphism, namely
 \[
 \xymatrix{
E \ar[d]_{P_G} \ar[rr]^{id} & & E \ar[d]^\pi  \\
E/G \ar[rr]^{f:\cong} & & B }
 \]
where $P_G$ is the usual projection. A $G$-bundle $(E, \pi, B, G)$ is principal if the action of $G$ on $E$ is free in the sense that $x g = x$ for some $x \in E$ and $g \in G$ implies $g=1,$ and the group $G$ is then called the structure group of the bundle $(E, \pi, B, G)$ (in physical literatures $G$ is also called the gauge group, cf. \cite{BMKNZ2003}). Note that in a principal $G$-bundle $(E, \pi, B, G)$ every fiber $\pi^{-1} (b)$ for $b \in B$ is homeomorphic to $G$ by the freedom of the $G$-action, hence it is a fiber bundle $(E, \pi, B, G)$ with the fiber $G$ and is simply called a principal fiber bundle with the structure group $G.$

Next, we construct principal fiber bundles over the observable space $\mathcal{W} (\mathbb{H}).$ To this end, fix a point $O_0 = \{ | e_n \rangle \langle e_n|: n \ge 1 \} \in \mathcal{W} (\mathbb{H}).$ For any $O \in \mathcal{W} (\mathbb{H}),$ we write
$$
\mathcal{F}^O_{O_0} = \{ U \in \mathcal{U} (\mathbb{H}):\; U^\dagger O U = O_0 \},
$$
that is, $U \in \mathcal{F}^O_{O_0}$ if and only if $\{ U| e_n \rangle: n \ge 1 \}$ is a basis such that $O = \{ U| e_n \rangle \langle e_n| U^*: n \ge 1 \}.$ Indeed, if $O = \{ |n \rangle \langle n|: n \ge 1 \},$ then
$$
\mathcal{F}^O_{O_0} = \mathcal{G} (U) = \bigg \{ \sum_{n \ge 1} e^{ \mathrm{i} \theta_n} | \sigma (n) \rangle \langle e_n |:\; \forall \sigma \in \Pi (d), \forall \theta_n \in [0, 2 \pi ) \bigg \},
$$
where $U$ is a unitary operator so that $|n\rangle = U |e_n\rangle$ for $n \ge 1.$ Also, define
\begin{equation}\label{eq:GaugeGroup}
\mathcal{G}_{O_0} = \bigg \{ \sum_{n \ge 1} e^{ \mathrm{i} \theta_n} |e_{\sigma (n)} \rangle \langle e_n |:\; \forall \sigma \in \Pi (d), \forall \theta_n \in [0, 2 \pi ) \bigg \}.
\end{equation}
We note that $\mathcal{U} (1)^d$ has a unitary representation
\begin{equation}\label{eq:UnitaryGroupRepresentation}
\mathcal{U} (1)^d = \bigg \{ U_\theta = \sum^d_{n = 1} e^{ \mathrm{i} \theta_n} |e_n \rangle \langle e_n |:\; \forall \theta = (\theta_1, \ldots, \theta_d) \in [0, 2 \pi )^d \bigg \},
\end{equation}
while $\Pi (d)$ has a unitary representation
\begin{equation}\label{eq:PermutationGroupRepresentation}
\Pi (d) = \bigg \{ U_\sigma = \sum^d_{n = 1} |e_{\sigma (n)} \rangle \langle e_n |:\; \forall \sigma \in \Pi (d) \bigg \}.
\end{equation}
Thus, $\mathcal{G}_{O_0}$ is a (non-abelian) subgroup of $\mathcal{U} (\mathbb{H})$ generated by $\mathcal{U} (1)^d$ and $\Pi (d).$

The (right) action of $\mathcal{G}_{O_0}$ on $\mathcal{F}^O_{O_0}$ is defined as: For any $G \in \mathcal{G}_{O_0},$
$$
(G, U) \mapsto U G
$$
for all $U \in \mathcal{F}^O_{O_0}.$ Evidently, this action is free and invariant, namely $\mathcal{F}^O_{O_0}\cdot G = \mathcal{F}^O_{O_0}$ for any $G \in \mathcal{G}_{O_0}$ and every $O \in \mathcal{W} (\mathbb{H}).$ Note that
$$
\mathcal{U} (\mathbb{H}) = \bigcup_{O \in \mathcal{W} (\mathbb{H})} \mathcal{F}^O_{O_0},
$$
and $\mathcal{F}^O_{O_0}$ is homeomorphic to $\mathcal{G}_{O_0}$ as topological spaces since $\mathcal{F}^O_{O_0} = \mathcal{G} [U]$ for some $U$ such that $O = \{ U| e_n \rangle \langle e_n| U^*: n \ge 1 \}.$

The following is then principal fiber bundles over the observable space.

\begin{definition}\label{df:PrincipalFiber}
Given $O_0 \in \mathcal{W} (\mathbb{H}),$ a principal fiber bundle over $\mathcal{W} (\mathbb{H})$ associated with $O_0$ is defined to be
$$
\xi_{O_0} (\mathbb{H}) = (\mathcal{U} (\mathbb{H}), \mathcal{W} (\mathbb{H}), \Pi_{O_0}, \mathcal{G}_{O_0}),
$$
where $\mathcal{U} (\mathbb{H})$ is the total space, and the bundle projection $\Pi_{O_0}: \mathcal{U} (\mathbb{H}) \mapsto \mathcal{W} (\mathbb{H})$ is defined by
$$
\Pi_{O_0} (U) = O
$$
provided $U \in \mathcal{F}^O_{O_0}$ for (unique) $O \in \mathcal{W} (\mathbb{H}),$ namely $\Pi^{-1} (O) = \mathcal{F}^O_{O_0}$ for every $O \in \mathcal{W} (\mathbb{H}).$

We simply denote this bundle by $\xi_{O_0} = \xi_{O_0} (\mathbb{H}).$
\end{definition}

\begin{remark}\rm
In the sequel, we will see that the fixed point $O_0 \in \mathcal{W} (\mathbb{H})$ physically plays the role of a complete measurement. Indeed, it corresponds to a von Neumann measurement. On the other hand, the point $O_0$ induces a differential structure over the base space $\mathcal{W} (\mathbb{H})$ and determines the geometric structure of $\xi_{O_0},$ namely quantum connection and parallel transportation.
\end{remark}

For any two points $O_0, O'_0 \in \mathcal{W} (\mathbb{H})$ with $O_0 = \{ | e_n \rangle \langle e_n|: n \ge 1 \}$ and $O'_0 = \{ | e_n' \rangle \langle e_n' |: n \ge 1 \},$ we define a unitary operator $U_0$ by $U_0 |e_n \rangle = | e_n' \rangle$ for $n \ge 1.$ Then the map $T: \xi_{O_0} \mapsto \xi_{O'_0}$ defined by $T U = U U^{-1}_0$ for all $U \in \mathcal{U} (\mathbb{H})$ is an isometric isomorphism on $\mathcal{U} (\mathbb{H})$ such that $T$ maps the fibers of $\xi_{O_0}$ onto the fibers of $\xi_{O'_0}$ over the same points in the base space $\mathcal{W} (\mathbb{H}),$ namely the following diagram is commutative:
 \[
 \xymatrix{
\mathcal{U} (\mathbb{H}) \ar[dr]_{\Pi_{O_0}} \ar[rr]^{T} & & \mathcal{U} (\mathbb{H}) \ar[dl]^{\Pi_{O'_0}}  \\
      &  \mathcal{W} (\mathbb{H}) &            }
 \]
that is, $\Pi_{O_0} = \Pi_{O'_0} \circ T.$ Thus, $\xi_{O_0}$ and $\xi_{O'_0}$ are isomorphic as principal fiber bundles (cf. \cite{Isham1999}).

\subsection{Quantum connection}\label{q-connection}

Since $d=\mathrm{dim} (\mathbb{H}) < \infty,$ $\xi_{O_0}$ has a natural differential structure, namely $\mathcal{U} (\mathbb{H})$ and $\mathcal{W} (\mathbb{H})$ are both differential manifolds. In fact, $\mathcal{U} (\mathbb{H})$ is a Lie group, while $\mathcal{W} (\mathbb{H})$ can be identified as a submanifold of the Grassmannian manifold of $d$-dimensional subspaces of $\mathcal{B} (\mathbb{H})$ (cf. \cite{ZV2018}). However, the classical differential structure over $\mathcal{W} (\mathbb{H})$ seems not to be applicable in the geometric interpretation of observable-geometric phases. To define the suitable concepts of quantum connection and parallel transportation over the principal fiber bundle $\xi_{O_0},$ we need to introduce a differential structure over $\mathcal{W} (\mathbb{H})$ associated with each fixed $O_0 \in \mathcal{W} (\mathbb{H}).$ Indeed, we will introduce a geometric structure over $\xi_{O_0}$ in a certain operator-theoretic sense, where the differential structure on $\mathcal{W} (\mathbb{H})$ is different from the one given in \cite{ZV2018}.

Let us begin with the definition of tangent vectors for $\mathcal{G}_{O_0}$ in the operator-theoretic sense.

\begin{definition}\label{df:q-tangvectorUgroup}
Fix $O_0 \in \mathcal{W} (\mathbb{H}).$ For a given $U \in \mathcal{G}_{O_0},$ an operator $Q \in \mathcal{B} (\mathbb{H})$ is called a tangent vector at $U$ for $\mathcal{G}_{O_0},$ if there is a curve $\chi: (-\varepsilon, \varepsilon) \ni t \mapsto U(t) \in \mathcal{G}_{O_0}$ with $\chi (0) = U$ such that for every $h \in \mathbb{H},$ the limit
$$
\lim_{t \to 0} \frac{U(t) (h) - U (h)}{t} = Q (h)
$$
in $\mathbb{H},$ denoted by $Q = \frac{d \chi (t)}{d t} \big |_{t=0}.$ The set of all tangent vectors at $U$ is denoted by $T_U \mathcal{G}_{O_0},$ and $T \mathcal{G}_{O_0}= \bigcup_{U \in \mathcal{G}_{O_0}} T_U \mathcal{G}_{O_0}.$ In particular, we denote $\mathrm{g}_{O_0} = T_U \mathcal{G}_{O_0}$ if $U =I.$
\end{definition}

It is easy to check that given $U \in \mathcal{G}_{O_0}$ with the form $U =\sum^d_{n = 1} e^{ \mathrm{i} \theta_n} |e_{\sigma (n)} \rangle \langle e_n |$ for some $\sigma \in \Pi (d),$ for every $Q \in T_U \mathcal{G}_{O_0}$ there exists a sequence of complex number $(\alpha_n)_{n \ge 1}$ such that
\begin{equation}\label{eq:VertVectStrucGroupExpress}
Q = \sum_{n \ge 1} \alpha_n |e_{\sigma (n)} \rangle \langle e_n |.
\end{equation}
In particular, each element $Q \in \mathrm{g}_{O_0}$ is of form
\begin{equation}\label{eq:VertVectLieAlg}
Q = \sum_{n \ge 1} \alpha_n |e_n \rangle \langle e_n |,
\end{equation}
where $(\alpha_n)_{n \ge 1}$ is a sequence of complex number. Thus, $T_U \mathcal{G}_{O_0}$ is a linear subspace of $\mathcal{B} (\mathbb{H}).$

The following is the tangent space for the base space $\mathcal{W} (\mathbb{H})$ in the operator-theoretic sense.

\begin{definition}\label{df:q-tangvectorBaseSpace}

\begin{enumerate}[{\rm 1)}]

\item Fix $O_0 \in \mathcal{W} (\mathbb{H}).$ A continuous curve $\chi: [a, b] \ni t \mapsto O(t) \in \mathcal{W} (\mathbb{H})$ is said to be differential at a fixed $t_0 \in (a, b)$ relative to $O_0,$ if there is a nonempty subset $\mathcal{A}$ of $\mathcal{B} (\mathbb{H})$ satisfying that for any $Q \in \mathcal{A}$ there exist $\varepsilon>0$ so that $(t_0 -\varepsilon, t_0 + \varepsilon) \subset [a,b]$ and a strongly continuous curve $\gamma: (t_0 -\varepsilon, t_0 + \varepsilon) \ni t \mapsto U_t \in \mathcal{F}^{O(t)}_{O_0}$ such that the limit
$$
\lim_{t \to t_0} \frac{U_t (h) - U_{t_0} (h)}{t - t_0} = Q (h)
$$
for any $h \in \mathbb{H}.$ In this case, $\mathcal{A}$ is called a tangent vector of $\chi$ at $t=t_0$ and denoted by
$$
\mathcal{A} = \frac{d O(t)}{d t}\big |_{t = t_0} = \frac{d \chi (t)}{d t} \big |_{t = t_0}.
$$
We can define the left (or, right) tangent vector of $\chi$ at $t = a$ (or, $t =b$) in the usual way.

\item Fix $O_0 \in \mathcal{W} (\mathbb{H}).$ Given $O \in \mathcal{W} (\mathbb{H}),$ a tangent vector of $\mathcal{W} (\mathbb{H})$ at $O$ relative to $O_0$ is define to be a nonempty subset $\mathcal{A}$ of $\mathcal{B} (\mathbb{H}),$ provided $\mathcal{A}$ is a tangent vector of some continuous curve $\chi$ at $t=0,$ where $\chi: (-\varepsilon, \varepsilon) \ni t \mapsto O(t) \in \mathcal{W} (\mathbb{H})$ with $\chi (0) = O,$ i.e., $\mathcal{A} = \frac{d O (t)}{d t} \big |_{t=0}.$ We denote by $T_O \mathcal{W} (\mathbb{H})$ the set of all tangent vectors at $O,$ and write $T \mathcal{W} (\mathbb{H}) = \bigcup_{O \in \mathcal{W} (\mathbb{H})} T_O \mathcal{W} (\mathbb{H}).$

\end{enumerate}
\end{definition}

Therefore, the tangent vectors for the base space $\mathcal{W} (\mathbb{H})$ is strongly dependent on the choice of a measurement point $O_0.$

\begin{definition}\label{df:q-tangvectorFiberSpace}

\begin{enumerate}[{\rm 1)}]

\item A strongly continuous curve $\gamma: [a, b] \ni t \mapsto U(t) \in \mathcal{U} (\mathbb{H})$ is said to be differential at a fixed $t_0 \in (a, b),$ if there is an operator $Q \in \mathcal{B} (\mathbb{H})$ such that the limit
$$
\lim_{t \to t_0} \frac{U_t (h) - U_{t_0} (h)}{t - t_0} = Q (h)
$$
for all $h \in \mathbb{H}.$ In this case, $Q$ is called the tangent vector of $\gamma$ at $t=t_0$ and denoted by
$$
Q = \frac{d \gamma (t)}{d t} \Big |_{t = t_0} = \frac{d U (t)}{d t} \Big |_{t = t_0}.
$$
We can define the left (or, right) tangent vector of $\gamma$ at $t = a$ (or, $t =b$) in the usual way.

Moreover, $\gamma$ is called a smooth curve, if $\gamma$ is differential at each point $t \in [a, b],$ and for any $h \in \mathbb{H},$ the $\mathbb{H}$-valued function $t \mapsto \frac{d \gamma (t)}{d t} (h)$ is continuous in $[a, b].$

\item Fix $O_0 \in \mathcal{W} (\mathbb{H}).$ For a given $P \in \mathcal{U} (\mathbb{H}),$ an operator $Q \in \mathcal{B} (\mathbb{H})$ is called a tangent vector of $\xi_{O_0}$ at $P,$ if there exists a strongly continuous curve $\gamma: (-\varepsilon, \varepsilon) \ni t \mapsto P_t \in \mathcal{U} (\mathbb{H})$ with $\gamma (0) = P,$ such that $\gamma$ is differential at $t=0$ and $Q = \frac{d \gamma (t)}{d t} \big |_{t =0}.$ Denote by $T_P \xi_{O_0} (\mathbb{H})$ the set of all tangent vectors of $\xi_{O_0}$ at $P$ relative to $O_0,$ and write
$$
T \xi_{O_0} (\mathbb{H}) = \bigcup_{P \in \mathcal{U} (\mathbb{H})} T_P \xi_{O_0} (\mathbb{H}).
$$

\item Fix $O_0 \in \mathcal{W} (\mathbb{H}).$ Given $P \in \mathcal{U} (\mathbb{H}),$ a tangent vector $Q \in T_P \xi_{O_0} (\mathbb{H})$ is said to be vertical, if there is a strongly continuous curve $\gamma: (-\varepsilon, \varepsilon)\ni t \mapsto P_t \in \mathcal{F}^{\Pi (P)}_{O_0}$ with $\gamma (0) = P$ such that $\gamma$ is differential at $t=0$ and $Q = \frac{d \gamma (t)}{d t} \big |_{t =0}.$ We denote by $V_P \xi_{O_0} (\mathbb{H})$ the set of all vertically tangent vectors at $P.$
\end{enumerate}
\end{definition}

\begin{remark}\rm
For any $P \in \mathcal{U} (\mathbb{H}),$ the tangent space $T_P \xi_{O_0} (\mathbb{H})$ is the same for all measurement points $O_0$ since it is the tangent space of $\mathcal{U} (\mathbb{H})$ at $P$ in the operator-theoretic sense. However, the vertically tangent space $V_P \xi_{O_0} (\mathbb{H})$ is different from each other for distinct measurement points. In particular, every $Q \in V_P \xi_{O_0} (\mathbb{H})$ with $O_0 = \{|e_n \rangle \langle e_n |: 1 \le n \le d \}$ has the form
\begin{equation}\label{eq:VertTangentVect}
Q = \sum_{n \ge 1} \alpha_n P |e_n \rangle \langle e_n |,
\end{equation}
where $(\alpha_n)_{n \ge 1}$ is a sequence of complex number.
\end{remark}

For any fixed $O_0 \in \mathcal{W} (\mathbb{H}),$ recall that for each $G \in \mathcal{G}_{O_0},$ the right action $R_G$ of $\mathcal{G}_{O_0}$ on $\xi_{O_0}$ is defined by
$$
R_G (U) = U G,\quad \forall U \in \mathcal{U} (\mathbb{H}).
$$
This induces a map $(R_G)_*: T_P \xi_{O_0} (\mathbb{H}) \mapsto T_{R_G(P)} \xi_{O_0} (\mathbb{H})$ for each $P \in \mathcal{U} (\mathbb{H})$ such that
$$
(R_G)_* (Q) = Q G,\quad \forall Q \in T_P \xi_{O_0} (\mathbb{H}).
$$
Since $R_G$ preserves the fibers of $\xi_{O_0},$ then $(R_G)_*$ maps $V_P \xi_{O_0} (\mathbb{H})$ into $V_{R_G(P)} \xi_{O_0} (\mathbb{H}).$

Now, we are ready to define the concept of quantum connection over the observable space.

\begin{definition}\label{df:q-connetion}
Fix $O_0 \in \mathcal{W} (\mathbb{H}).$ A connection on the principal fiber bundle $\xi_{O_0}= (\mathcal{U} (\mathbb{H}), \mathcal{W} (\mathbb{H}), \Pi_{O_0}, \mathcal{G}_{O_0})$ is a family of linear functionals $\Omega = \{\Omega_P:\; P \in \mathcal{U} (\mathbb{H}) \},$ where for each $P \in \mathcal{U} (\mathbb{H}),$ $\Omega_P$ is a linear functional in $T_P \xi_{O_0} (\mathbb{H})$ with values in $\mathrm{g}_{O_0},$ satisfying the following conditions:
\begin{enumerate}[{\rm (1)}]

\item For any $P \in \mathcal{U} (\mathbb{H})$ and for all vertically tangent vectors $Q \in V_P \xi_{O_0} (\mathbb{H}),$ one has
\begin{equation}\label{eq:q-ConnectionVertTangVect}
\Omega_P (Q) = P^{-1} Q.
\end{equation}

\item $\Omega_P$ depends continuously on $P,$ in the sense that if $P_n$ converges to $P$ as well as $Q_n \in T_{P_n} \xi_{O_0} (\mathbb{H})$ converges $Q_0 \in T_P \xi_{O_0} (\mathbb{H})$ in the operator topology of $\mathcal{B} (\mathbb{H}),$ then $\lim_{n \to \infty} \Omega_{P_n} (Q_n) = \Omega_P (Q_0)$ in $\mathrm{g}_{O_0}.$

\item Under the right action of $\mathcal{G}_{O_0}$ on $\xi_{O_0} (\mathbb{H}),$ $\Omega$ transforms according to
\begin{equation}\label{eq:GaugeTransConnection}
\Omega_{R_G(P)} [(R_G)_* (Q )] = G^{-1} \Omega_P (Q) G,
\end{equation}
for $G \in \mathcal{G}_{O_0},$ $P \in \mathcal{U} (\mathbb{H}),$ and $Q \in T_P \xi_{O_0} (\mathbb{H}).$

\end{enumerate}
Such a connection is simply called an $O_0$-connection.
\end{definition}

Next, we present a canonical example of such quantum connections, which plays a crucial role in the expression of observable-geometric phases.

\begin{example}\label{Ex:CanonicalConnection}\rm
Given a fixed $O_0  = \{|e_n \rangle\langle e_n|: 1 \le n \le d \} \in \mathcal{W} (\mathbb{H}),$ we define $\check{\Omega} = \{\check{\Omega}_P: P \in \mathcal{U} (\mathbb{H}) \}$ as follows: For each $P \in \mathcal{U} (\mathbb{H}),$ $\check{\Omega}_P : T_P \xi_{O_0} (\mathbb{H}) \mapsto \mathrm{g}_{O_0}$ is defined by
\begin{equation}\label{eq:CanonConnction}
\check{\Omega}_P (Q) = P^{-1} \star Q
\end{equation}
for any $Q \in T_P \xi_{O_0} (\mathbb{H}),$ where
$$
P^{-1} \star Q = \sum^d_{n = 1}  \langle e_n | P^{-1} Q | e_n \rangle |e_n \rangle \langle e_n|.
$$
By \eqref{eq:VertTangentVect}, one has $P^{-1} \star Q = P^{-1} Q \in \mathrm{g}_{O_0}$ for any $Q \in V_P \xi_{O_0} (\mathbb{H}),$ namely $\check{\Omega}_P$ satisfies \eqref{eq:q-ConnectionVertTangVect}. The conditions (2) and (3) of Definition \ref{df:q-connetion} are clearly satisfied by $\check{\Omega}.$ Hence, $\check{\Omega}$ is an $O_0$-connection on $\xi_{O_0}.$ In this case, we write $\check{\Omega}_P = P^{-1} \star d P$ for any $P \in \mathcal{U} (\mathbb{H}).$
\end{example}

\subsection{Quantum parallel transportation}\label{q-ParallelTransport}

The quantum parallel transportation in the state space was introduced in \cite{Simon1983, AA1987} and studied in \cite{Anandan1992} in details. This section is devoted to the study of quantum parallel transport over the observable space.

\begin{definition}\label{df:q-lift}
Fix a point $O_0 \in \mathcal{W} (\mathbb{H}).$ For a continuous curve $C_W: [a, b] \ni t \longmapsto O (t) \in \mathcal{W} (\mathbb{H}),$ a lift of $C_W$ with respect to $O_0$ is defined to be a continuous curve
$$
C_P: [a, b] \ni t \longmapsto U(t) \in \mathcal{U} (\mathbb{H})
$$
such that $U(t) \in \mathcal{F}^{O(t)}_{O_0}$ for any $t\in [a, b].$
\end{definition}

\begin{remark}\label{rk:q-lift}\rm
Note that, a lift of $C_W$ depends on the choice of the point $O_0;$ for the same curve $C_W,$ lifts are distinct for different points $O_0.$ For this reason, such a lift $C_P$ is called a $O_0$-lift of $C_W.$
\end{remark}

\begin{definition}\label{df:SmoothCurve}
Fix a point $O_0 \in \mathcal{W} (\mathbb{H}).$ A continuous curve $C_W: [a, b] \ni t \longmapsto O (t) \in \mathcal{W} (\mathbb{H})$ is said to be smooth, if it has a $O_0$-lift $C_P: [a, b] \ni t \longmapsto U(t) \in \mathcal{U} (\mathbb{H})$ which is a smooth curve. In this case, $C_P$ is called a smooth $O_0$-lift of $C_W.$
\end{definition}

Note that, if a continuous curve $C_W: [a, b] \ni t \longmapsto O (t) \in \mathcal{W} (\mathbb{H})$ is smooth, then it is differential at every point $t \in [a, b].$ Indeed, suppose that $C_P: [a, b] \ni t \longmapsto U(t) \in \mathcal{U} (\mathbb{H})$ is a smooth $O_0$-lift of $C_W.$ For each $t \in [a, b],$ we have $\frac{d C_P (t)}{d t} \in \frac{d O(t)}{d t},$ namely $\frac{d O(t)}{d t}$ is a nonempty subset of $\mathcal{B} (\mathbb{H}),$ and hence $C_W$ is differential.

\begin{definition}\label{df:q-HorizontalLift}
Fix $O_0 \in \mathcal{W} (\mathbb{H})$ and let $\Omega$ be an $O_0$-connection on $\xi_{O_0} (\mathbb{H}).$ Let $C_W: [0, T] \ni t \longmapsto O (t) \in \mathcal{W} (\mathbb{H})$ be a smooth curve. If $C_P: [0, T] \ni t \longmapsto \tilde{U} (t)$ is a smooth $O_0$-lift of $C_W$ such that
\begin{equation}\label{eq:ParallelTransfConnectionCond}
\Omega_{\tilde{U} (t)} \Big [ \frac{d \tilde{U} (t)}{d t}\Big ] =0
\end{equation}
for every $t \in [0, T],$ then $C_P$ is called a horizontal lift of $C_W$ with respect to $O_0$ and $\Omega.$

In this case, $C_P$ is simply called the horizontal $O_0$-lift of $C_W$ associated with $\Omega.$ And, the curve $C_P: t \mapsto \tilde{U} (t)$ is called the parallel transportation along $C_W$ with the starting point $C_P (0) = \tilde{U}(0)$ with respect to the connection $\Omega$ on $\xi_{O_0} (\mathbb{H}).$
\end{definition}

The following proposition shows the existence of the horizontal lifts.

\begin{proposition}\label{prop:ParallelTransf}\rm
Fix $O_0 \in \mathcal{W} (\mathbb{H})$ and let $\Omega$ be an $O_0$-connection on $\xi_{O_0} (\mathbb{H}).$ Let $C_W: [0, T] \ni t \longmapsto O (t) \in \mathcal{W} (\mathbb{H})$ be a smooth curve. For any $U_0 \in \mathcal{F}^{O(0)}_{O_0},$ there exists a unique horizontal $O_0$-lift $\tilde{C}_P$ of $C_W$ with the initial point $\tilde{C}_P (0) = U_0.$
\end{proposition}

\begin{proof}
Let $\Gamma: [0, T] \ni t \longmapsto U (t) \in \mathcal{U} (\mathbb{H})$ be a smooth $O_0$-lift of $C_W$ with $\Gamma (0) = U_0.$ To prove the existence, note that the condition (2) of Definition \ref{df:q-connetion} implies that the function $t \mapsto \Omega_{\Gamma (t)} \big [ \frac{d \Gamma (t)}{d t} \big ]$ is continuous in $[0, T].$ Then,
\begin{equation}\label{eq:GeodesicEquaGaugeTransf}
\frac{d G (t)}{d t} = - \Omega_{\Gamma (t)} \Big [ \frac{d \Gamma (t)}{d t} \Big ] \cdot G (t)
\end{equation}
with $G (0) = I$ has the unique solution in $[0, T].$ Therefore, $\tilde{C}_P (t) = \Gamma (t) \cdot G (t)$ is the required horizontal $O_0$-lift of $C_W$ for the initial point $U_0 \in \mathcal{F}^{O(0)}_{O_0}.$

To prove the uniqueness, suppose $\check{C}_P: [0, T] \ni t \longmapsto \check{U} (t) \in\mathcal{U} (\mathbb{H})$ be another horizontal $O_0$-lift of $C_W$ for the initial point $U \in \mathcal{F}^{O(0)}_{O_0}.$ Then $\check{C}_P (t) = \tilde{C}_P (t) \cdot \check{G} (t)$ for all $t \in [0, T],$ where $\check{G}(0) = I.$ Since
$$
0 = \Omega_{\check{U}(t)} \Big [ \frac{d \check{U}(t)}{d t} \Big ] = \check{G}(t)^{-1} \frac{d \check{G}(t)}{d t},
$$
this follows that $\check{G}(t) = I$ for all $t \in [0, T].$ Hence, the horizontal $O_0$-lift of $C_W$ is unique for the initial point $U \in \mathcal{F}^{O(0)}_{O_0}.$
\end{proof}

\begin{example}\label{Ex:QuantumParallelTransport}\rm
Let $C_P: [0,T] \ni t \mapsto U(t) \in \mathcal{U} (\mathbb{H})$ be a unitary evolution satisfying the Schr\"{o}dinger equation
\begin{equation}\label{eq:SchrodingerEquTimeUnitaryEvolution}
\mathrm{i} \frac{d U(t)}{d t} = H(t) U (t)
\end{equation}
where $H(t)$ are time-dependent Hamiltonian operators in $\mathbb{H}.$ Given a fixed $O_0 = \{|e_n \rangle\langle e_n|: 1 \le n \le d \} \in \mathcal{W} (\mathbb{H}),$ define $C_W: [0, T] \ni t \longmapsto O (t) \in \mathcal{W} (\mathbb{H})$ by $O(t) = U(t) O_0 U(t)^{-1}$ for all $t \in [0, T].$ We define $\tilde{C}_P: [0,T] \ni t \mapsto \tilde{U}(t) \in \mathcal{U} (\mathbb{H})$ by
$$
\tilde{U} (t) = \sum^d_{n=1} \exp \Big ( \mathrm{i} \int^t_0 \langle e_n | \Big [ U(t')^{-1} \frac{d U(t')}{d t'} \Big ] | e_n \rangle d t' \Big ) U(t) | e_n \rangle \langle e_n|
$$
for every $t \in [0, T],$ along with the initial point $\tilde{U}(0) = U(0) \in \mathcal{F}^{O(0)}_{O_0}.$ Then $\tilde{C}_P$ is a smooth $O_0$-lift of $C_W$ such that
$$
\check{\Omega}_{\tilde{U} (t)} \Big [ \frac{d \tilde{U}(t)}{d t} \Big ] =0
$$
for all $t \in [0, T],$ where $\check{\Omega}$ is the canonical $O_0$-connection introduced in Example \ref{Ex:CanonicalConnection}. Thus, $\tilde{C}_P$ is the horizontal $O_0$-lift of $C_W$ with respect to $\check{\Omega},$ namely $\tilde{C}_p$ is the parallel transportation along $C_W$ with the starting point $C_P (0) = U(0)$ with respect to the connection $\check{\Omega}$ on $\xi_{O_0} (\mathbb{H}).$
\end{example}

\subsection{Geometric interpretation of observable-geometric phases}\label{GeoInter}

We are now ready to give a geometric interpretation of $\beta_n$'s defined as in \eqref{eq:q-GeoPhase} in Section \ref{ObGeoPhase}. Indeed, given a point $O_0 = \{ | e_n \rangle \langle e_n|: n \ge 1 \} \in \mathcal{W} (\mathbb{H}),$ using the notations involved in Section \ref{ObGeoPhase}, we define $\tilde{U} (t) \in \mathcal{U} (\mathbb{H})$ for $0 \le t \le T$ by
$$
\tilde{U} (t) = \sum^d_{n = 1} |\tilde{\psi}_n (t) \rangle \langle e_n |,
$$
where $|\tilde{\psi}_n (t) \rangle$'s are defined in \eqref{eq:Parallelvect}. Then,
$$
\tilde{C}_P: \; [0, T] \ni t \longmapsto \tilde{U} (t) \in \mathcal{U} (\mathbb{H})
$$
is a smooth $O_0$-lift of $C_W: [0, T] \ni t \mapsto O(t) = \{ |\psi_n (t) \rangle \langle \psi_n (t)|: 1 \le n \le d \}.$ By \eqref{eq:ParallelCondVect}, we have
\begin{equation}\label{eq:CanonicalParallelCond}
\check{\Omega}_{\tilde{U} (t)} \Big [ \frac{d \tilde{U} (t)}{d t}  \Big ] = 0
\end{equation}
for all $t \in [0,T],$ where $\check{\Omega}$ is the canonical connection (cf. Example \ref{Ex:CanonicalConnection}). This means that $[0, T] \ni t \mapsto \tilde{U} (t)$ is the parallel transportation along $C_W$ associated with the {\it canonical connection} $\check{\Omega}$ on $\xi_{O_0}.$ Therefore,  $\tilde{C}_P$ is the {\it horizontal} $O_0$-lift of $C_W$ with respect to $\check{\Omega}$ in the principal bundle $\xi_{O_0}$ such that $\tilde{U} (T) |e_n \rangle = |\tilde{\psi}_n (T) \rangle = e^{\mathrm{i} \beta_n} |\psi_n \rangle,$ and so
\begin{equation}\label{eq:HolonomyUnitaryOper}
\tilde{U} (T) = \sum^d_{n = 1} e^{\mathrm{i} \beta_n} | \psi_n \rangle \langle e_n |
\end{equation}
is the holonomy element associated with the connection $\check{\Omega}, C_W,$ and $U_0 = \sum^d_{n = 1}| \psi_n \rangle \langle e_n |$ in $\xi_{O_0}.$

In conclusion, we have the following theorem.

\begin{theorem}\label{thm:q-GeoPhase}\rm
\begin{enumerate}[\rm (1)]

\item For every $n = 1,\ldots, d,$ the geometric phase $\beta_n$ defined in \eqref{eq:q-GeoPhase} is given by
\begin{equation}\label{eq:q-GeoPhaseExpression}
\beta_n =\langle e_n | \mathrm{i} \int^T_0 \check{\Omega}_{U (t)} \Big [ \frac{d U (t)}{d t} \Big ] d t |e_n \rangle = \langle e_n | \mathrm{i} \oint_{C_W} U^{-1} \star d U |e_n \rangle,
\end{equation}
where $C_P: [0, T] \ni t \mapsto U (t) \in \mathcal{U} (\mathbb{H})$ corresponds to any of the closed smooth $O_0$-lifts of $C_W$ with $U(0) = U_0,$ and $\check{\Omega}_U = U^{-1} \star d U$ is the canonical connection on $\xi_{O_0} (\mathbb{H}).$ Thus, $\beta_n$'s are independent of the choice of the time parameterization of $U(t),$ namely the speed with which $U(t)$ traverses its closed path. It is also independent of the choice of the Hamiltonian as long as the Heisenberg equations \eqref{eq:HeisenbergEquTime} involving these Hamiltonians describe the same closed path $C_W$ in $\mathcal{W} (\mathbb{H}).$

\item The set $\{ \beta_n: 1 \le n \le d\}$ is independent of the choice of the starting point $U_0.$

\item The set $\{ \beta_n: 1 \le n \le d\}$ is independent of the choice of the measurement point $O_0.$ Therefore, this number set is considered to be a set of geometric invariants for $C_W.$

\end{enumerate}
\end{theorem}

\begin{remark}\label{rk:q-ObsGeoPhase}\rm
In the definition \eqref{eq:q-GeoPhase}, the $\beta_n$'s are in fact independent of the choice of measurement points $O_0$ and background geometry over the fiber bundle $\xi_{O_0}.$
\end{remark}

\begin{proof}
(1).\; Let $C_P: [0, T] \ni t \longmapsto U (t) \in \mathcal{F}^{O(t)}_{O_0}$ be a smooth $O_0$-lift of $C_W$ such that $U(T) = U(0)=U_0.$ By Proposition \ref{prop:ParallelTransf} we have $\tilde{U} (t) = U(t) G(t)$ for every $0 \le t \le T,$ where $G(t)$ satisfies \eqref{eq:GeodesicEquaGaugeTransf} with $\Gamma (t) = U(t).$ Then there are continuously differential functions $\psi_n$ with $\psi_n(0)=0$ and $\psi_n (T) = \beta_n$ such that $G(t) | e_n \rangle = e^{\mathrm{i} \psi_n (t)} | e_n \rangle$ for all $n \ge 1.$ By \eqref{eq:q-ConnectionVertTangVect} and \eqref{eq:GaugeTransConnection}, one has
$$
\int^T_0 \langle e_n | \check{\Omega}_{U(t)} \Big [ \frac{d U(t)}{d t}\Big ] | e_n \rangle d t= \int^T_0 \langle e_n | G(t) \frac{d G(t)^{-1}}{d t} | e_n \rangle d t= - \mathrm{i} \beta_n.
$$
This follows \eqref{eq:q-GeoPhaseExpression}.

(2).\; For any $\check{U}_0 \in \mathcal{F}^{O(0)}_{O_0}$ there exists some $G = \sum^d_{n=1} e^{\mathrm{i} \theta_n}| e_{\sigma (n)} \rangle \langle e_n | \in \mathcal{G}_{O_0}$ with $\sigma \in \Pi (d)$ and $\theta_n \in \mathbb{R}$ such that $\check{U}_0 = U_0 G.$ Then $\check{C}_P: [0, T] \ni t \longmapsto \check{U}(t) = \tilde{U} (t) G$ is the horizontal $O_0$-lift of $C_W$ with the starting point $\check{U}(0) = U_0 G$ such that $\check{U}(T) |e_n\rangle = e^{\mathrm{i} \beta_{\sigma (n)}} \check{U} (0) |e_n\rangle$ for all $1 \le n \le d.$ Thus, the set $\{\beta_n: n \ge 1 \}$ is invariant for any starting point $U_0 \in \mathcal{F}^{O(0)}_{O_0}.$ Combining this fact with \eqref{eq:q-GeoPhaseExpression} yields
$$
\{\beta_n: n \ge 1 \} = \Big \{ \mathrm{i} \langle e_n | \oint_{\bar{C}_P} \check{\Omega}_{\bar{U}} [ d \bar{U} ] | e_n \rangle:\; n=1, \ldots, d \Big \}
$$
for any closed smooth $O_0$-lift $\bar{C}_P$ of $C_W.$ Therefore, the observable-geometric phases are independent of the choice of the starting point and only depends on the geometry of the curve $C_W$ with respect to the $O_0$-connection $\check{\Omega}.$

(3).\; Let $\tilde{C}_P: [0, T] \ni t \longmapsto \tilde{U} (t)$ be the horizontal $O_0$-lift of $C_W$ with respect to $\Omega$ with the starting point $\tilde{U} (0) = U_0.$ For any $O'_0 = \{ | e_n' \rangle \langle e_n'|: 1 \le n \le d \} \in \mathcal{W} (\mathbb{H})$ there exists some $U' \in \mathcal{U} (\mathbb{H})$ such that $O'_0= U' O_0 U'^{-1}$ with $| e_n' \rangle = U' | e_n \rangle$ for $n= 1, \ldots, d.$ Then $\Omega' = \{\Omega'_P: P \in \mathcal{U} (\mathbb{H}) \}$ is a $O'_0$-connection on $\xi_{O'_0},$ where $\Omega'_P (Q) = U' \Omega_{P U'} (Q U') U'^{-1}$ for any $P \in \mathcal{U} (\mathbb{H})$ and for all $Q \in T_P \xi_{O'_0} (\mathbb{H}).$ By computation, we conclude that $\tilde{C}'_P: [0, T] \ni t \longmapsto \tilde{U}' (t) = \tilde{U} (t) U'^{-1}$ is the horizontal $O'_0$-lift of $C_W$ with respect to $\Omega'$ with the starting point $\tilde{U}' (0) = U_0 U'^{-1}.$ Therefore,
$$
\tilde{U}' (T) | e_n' \rangle = \tilde{U} (T) U'^{-1} | e_n' \rangle = \tilde{U} (T) | e_n \rangle = e^{\mathrm{i} \beta_n} \tilde{U} (0) | e_n \rangle = e^{\mathrm{i} \beta_n} \tilde{U}' (0) | e_n' \rangle,
$$
and hence the set of the geometric phases of $C_W$ with respect to $\Omega'$ is the same as that of $\Omega.$
\end{proof}

\

{\it Acknowledgments}\; This work is partially supported by the Natural Science Foundation of China under Grant No.11871468, and also in part supported by MOST under Grant No. 2017YFA0304500.

\bibliography{apssamp}

\end{document}